\documentclass[conference,letterpaper]{IEEEtran}
\usepackage[utf8]{inputenc} 
\usepackage[T1]{fontenc}
\usepackage{url}
\usepackage{ifthen}
\usepackage{cite}
\usepackage[cmex10]{amsmath}
\usepackage{graphicx}
\usepackage{amssymb}
\usepackage{amsthm}
\usepackage{mathtools}
\usepackage{hyperref}
\usepackage{bbm}
\usepackage{tikz}
\usepackage[most]{tcolorbox}
\usepackage{algorithm}
\usepackage{algpseudocode}

\usetikzlibrary{patterns,calc,intersections}
\usetikzlibrary{shapes,shapes.misc,arrows.meta,calc}

\newtheorem{theorem}{Theorem}
\theoremstyle{definition}

\newtheorem{lemma}{Lemma}
\newtheorem{example}{Example}

\begin{document}

\title{Learning from Acceptance: Cumulative Regret in the Game of Coding}

\author{
  \IEEEauthorblockN{Hanzaleh Akbari Nodehi, Parsa Moradi, and Mohammad Ali Maddah-Ali}
  \IEEEauthorblockA{University of Minnesota, Twin Cities, Minneapolis, MN, USA}
}
\maketitle

\begin{abstract}
Classical coding-theoretic guarantees often rely on trust assumptions, such as requiring sufficiently many honest nodes compared with adversarial ones. These assumptions are difficult to enforce in open decentralized systems where participants are not centrally certified. At the same time, such environments often contain incentive mechanisms: participants may be rewarded only when their submitted data are accepted and the system remains functional. This changes the role of an adversary. Rather than acting as a pure saboteur, a strategic adversary may submit data that are consistent enough to be accepted while still degrading the quality of the final estimate.

The game-of-coding framework models this strategic interaction between a data collector (DC) and an adversary. Existing works on the game of coding mostly consider the complete-information case, where the DC knows how the adversary trades off acceptance and estimation error. In this paper, we study an incomplete-information version of the game of coding in which the DC, acting as a Stackelberg leader, does not know the adversary's utility trade-off and must learn through repeated interaction. Prior work on the unknown-adversary setting considered an explore-then-commit objective, where only the final selected acceptance rule is evaluated. In contrast, we study the full learning trajectory: every acceptance rule used during the algorithm is executed and contributes to performance. We propose an algorithm that refines its search around promising acceptance rules, prove that it achieves sublinear cumulative regret, and evaluate its performance through numerical experiments.
\end{abstract}

\section{Introduction}
\label{sec:introduction}

Classical coding theory provides fundamental tools for reliable communication, storage, and computation, but it relies on strict trust assumptions. Successful recovery is typically guaranteed only when honest nodes sufficiently outnumber adversarial ones. For example, in a system with a set of honest nodes $\mathcal{H}$ and adversarial nodes $\mathcal{T}$, standard repetition coding requires $|\mathcal{H}| \geq |\mathcal{T}| + 1$ to guarantee recovery. Similarly, Reed--Solomon codes, which encode $K$ symbols into $N$ coded symbols, require $|\mathcal{H}| \geq |\mathcal{T}| + K$~\cite{SudanBook}. Comparable threshold requirements also appear in coded computation, as well as analog and real-valued settings~\cite{ZamirCoded, jahani2018codedsketch, roth2020analog, moradi2025general, dutta2019optimal, dutta2016short, yu2017polynomial, yu2020entangled, yu2019lagrange}.

These trust assumptions are difficult to enforce in open decentralized systems where participants are not centrally certified~\cite{sliwinski2019blockchains, han2021fact, gans2023zero}. This issue is particularly relevant in blockchain-based platforms~\cite{bitcoin2008bitcoin, buterin2013ethereum, ruoti2019sok}, including decentralized machine learning and oracle systems~\cite{shafay2023blockchain, ding2022survey, kayikci2024blockchain, tian2022blockchain, salah2019blockchain}.  However, these platforms are not purely worst-case environments. Participants are incentive-driven and may receive rewards only when their contribution is accepted and the system remains functional. Therefore, an adversarial participant may not wish to destroy recoverability completely. Instead, it may try to keep the system live while still degrading the quality of the estimate or computation. This incentive structure suggests that adversarial behavior in these settings should be modeled strategically, rather than purely as arbitrary corruption.

To capture this strategic behavior, recent literature in coding theory introduced the game-of-coding \cite{GoCJournal, nodehi2026gamevector, GoDSybil, nodehi2026game, nodehi2025unknown}. In this framework, a data collector (DC) interacts with an adversary, and the utility of each player depends on two key quantities: the probability that the DC accepts the submitted reports and the estimation error conditioned on acceptance. This creates a trade-off. The adversary wants to increase the error, but it must avoid causing rejection and losing its reward. Conversely, the DC wants to reduce the  error while maintaining enough acceptance probability to keep the system live.

A basic two-node model of the game of coding is illustrated in Fig.~\ref{fig:two_node_model}. The system consists of one honest node, one adversarial node, and a DC. The DC aims to estimate an unknown random variable $\mathbf{u}$, but it does not observe $\mathbf{u}$ directly. Instead, it receives two reported values. The honest node sends $\mathbf{y}_h=\mathbf{u}+\mathbf{n}_h$, where the honest noise is bounded by $\Delta$, while the adversarial node sends $\mathbf{y}_a=\mathbf{u}+\mathbf{n}_a$, where the distribution of $\mathbf{n}_a$ is strategically chosen by the adversary. The DC receives $\underline{\mathbf{y}}=\{\mathbf{y}_1,\mathbf{y}_2\}$ and does not know which report is honest.

A natural consistency check accepts the reports only if $|\mathbf{y}_1-\mathbf{y}_2|\leq 2\Delta$, since any two honest reports would satisfy this bound. However, this rule can be overly conservative. Reports outside this narrow region may still yield a sufficiently accurate estimate of $\mathbf{u}$, and rejecting them may unnecessarily reduce system functionality. Therefore, the DC may instead use a more flexible acceptance rule of the form $|\mathbf{y}_1-\mathbf{y}_2|\leq \eta\Delta$. In this setting, the DC acts as a Stackelberg leader by committing to the threshold $\eta$, while the adversary observes $\eta$ and strategically chooses its noise distribution in response. The equilibrium and optimal strategies for this game were first analyzed in~\cite{GoCJournal}, and subsequent works extended the framework to vector-valued computations, multi-node systems, Sybil attacks, and advanced coding techniques~\cite{nodehi2026gamevector, GoDSybil, nodehi2026game}.

The works above consider a complete-information game, where the DC and the adversary know each other's utility. This assumption is restrictive in many practical settings. It is reasonable to assume that the adversary know the DC's utility, but the DC may not know how the adversary trades off error and probability of acceptance. As a result, the DC may not know the adversary's best response to a committed threshold $\eta$. This leads to an incomplete-information game, where the DC must learn the effect of its committed thresholds through repeated interaction. Learning in games with unknown follower utilities has been studied in several settings~\cite{conitzer2006computing, gan2023robust, letchford2009learning, blum2014learning, peng2019learning, sessa2020learning, balcan2015commitment, haghtalab2022learning, dong2018strategic, kleinberg2003value}.

The unknown-adversary version of the game of coding was first studied in~\cite{nodehi2025unknown}. In that formulation, the DC interacts with the adversary for several rounds and is evaluated only through the final selected threshold. This explore-then-commit objective is appropriate when the learning phase is separated from the deployment phase. In many decentralized systems, however, this separation is not available. Every threshold used during learning is actually executed, and each accepted or rejected report affects the system's performance. Motivated by this observation, we study a trajectory-level objective in which every round contributes to the DC's performance. We measure performance by cumulative regret with respect to the best fixed threshold in hindsight.

This  objective connects the problem to the multi-armed bandit framework, where a learner repeatedly chooses actions, observes feedback, and aims to minimize regret relative to the best fixed action~\cite{LaiRobbins1985, AuerCesaBianchiFischer2002, BubeckCesaBianchi2012, LattimoreSzepesvari2020, Slivkins2019}. Since the DC's threshold $\eta$ belongs to a continuum, our setting is especially related to continuum-armed and metric-space bandits~\cite{Agrawal1995continuum, Kleinberg2004continuum, AuerOrtnerSzepesvari2007continuum, KleinbergSlivkinsUpfal2008metric, BubeckMunosStoltzSzepesvari2011}. Nevertheless, the feedback in our problem is different from standard bandit. In a classical bandit problem, after choosing an action, the learner observes a reward sample and can use it to evaluate that action. In our setting, after the DC commits to a threshold, the adversary strategically responds, and the DC observes only whether the submitted reports were accepted or rejected. The DC does not observe its utility, nor does it observe the estimation error in that round, because the true value of $\mathbf{u}$ is unknown. Thus, the feedback is both partial and strategic: the observed binary acceptance event is generated after the adversary has responded to the DC's committed threshold.

To address this challenge, we exploit a structural property of the game of coding. For any threshold committed by the DC and any corresponding response of the adversary, the induced probability of acceptance and the estimation error are not arbitrary. They lie on a curve determined by the statistical model and the acceptance rule. Importantly, this curve is independent of the utility functions of both players. Therefore, even though the DC cannot observe the estimation error directly, it can use repeated accept/reject observations to estimate the probability of acceptance and then map this estimate to the corresponding induced utility.

Using this idea, we design an adaptive threshold-learning algorithm. The algorithm maintains a set of active thresholds and updates an estimate of the DC's utility for each active threshold using only binary acceptance feedback. It assigns larger uncertainty bonuses to thresholds that have been sampled less often and gradually refines the search by activating new thresholds in regions that are not yet well represented. As a result, the algorithm spends more time near thresholds that may be close to optimal while still exploring insufficiently sampled regions of the threshold space.

The contributions of this paper are as follows. First, we formulate the unknown-adversary game of coding as an online Stackelberg learning problem with a cumulative regret objective, where every threshold used during learning affects performance. Second, we show how the DC can learn from binary
acceptance feedback, even though it does not directly observe its utility. Third, we design an adaptive threshold-learning policy for the continuous threshold space and prove that it achieves sublinear regret. Finally, we provide numerical experiments showing that the proposed policy outperforms the explore-then-commit baseline of~\cite{nodehi2025unknown}.

The paper is organized as follows. Section~\ref{sec:problem-formulation} presents the problem formulation. Section~\ref{sec:preliminaries} provides the preliminaries. Section~\ref{sec:red-main-result} presents the main theorem, which is proved in Section~\ref{sec:red-proof-main-result}. Section~\ref{sec:experiments} presents the numerical experiments.

\section{Problem Formulation}
\label{sec:problem-formulation}

In this section we present the problem formulation. We consider a system consisting of two nodes\footnote{For simplicity of exposition, in this conference paper we focus on system with two nodes. Based on \cite{GoDSybil}, the result of this paper can be directly used for system with more than two nodes.} and a DC (See Fig.~\ref{fig:two_node_model}). The DC aims to estimate $\mathbf{u}$ but does not have direct access to it. Instead, it relies on the two nodes. One of the nodes is honest, while the other is adversarial. The identity of the adversarial node is hidden from the DC. The honest node sends $\mathbf{y}_h=\mathbf{u}+\mathbf{n}_h$, where the probability density function (PDF) of $\mathbf{n}_h$, denoted by $f_{\mathbf{n}_h}$, is symmetric, and $\Pr\left(|\mathbf{n}_h|>\Delta\right)=0$ for some $\Delta\in\mathbb{R}$. The adversary sends $\mathbf{y}_a=\mathbf{u}+\mathbf{n}_a$, where $\mathbf{n}_a$ is independent of $\mathbf{u}$. The PDF of $\mathbf{n}_a$, denoted by $g(.)$, is chosen by the adversary. The DC is unaware of this choice. We assume that $M$, $\Delta$, and the distribution of $\mathbf{n}_h$ are known to all players, and $\Delta\ll M$.

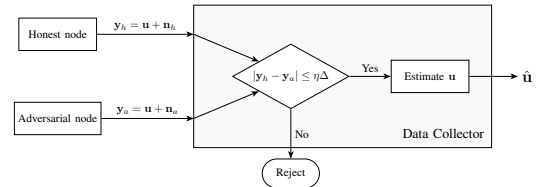
\begin{figure}[htbp]
    \centering
    \resizebox{0.78\columnwidth}{!}{%
        \begin{tikzpicture}[>=Stealth, thick]

            \tikzset{
                block/.style={draw, rectangle, minimum height=10mm, minimum width=24mm, align=center, fill=white},
                diamondblock/.style={draw, diamond, aspect=1.7, minimum width=36mm, minimum height=20mm, align=center, fill=white, inner sep=1pt},
                rejectblock/.style={draw, rounded rectangle, minimum height=9mm, minimum width=20mm, align=center, fill=white}
            }

            \def\boxL{0}
            \def\boxR{9.2}
            \def\boxT{2.2}
            \def\boxB{-2.2}

            \def\nodeX{-4.2}
            \def\nodeY{1.3}

            \def\diaX{3.0}
            \def\estX{7.3}
            \def\rejY{-3.0}

            \draw[thick, fill=gray!5] (\boxL,\boxT) rectangle (\boxR,\boxB);
            \node[anchor=south east, font=\large] at (\boxR-0.1,\boxB+0.15) {Data Collector};

            \node[block] (nodeh) at (\nodeX,\nodeY) {Honest node};
            \node[block] (nodea) at (\nodeX,-\nodeY) {Adversarial node};

            \node[diamondblock] (decision) at (\diaX,0) {$|\mathbf{y}_h-\mathbf{y}_a|\leq \eta\Delta$};

            \node[block] (estimation) at (\estX,0) {Estimate $\mathbf{u}$};
            \node[rejectblock] (reject) at (\diaX,\rejY) {Reject};

            \node[font=\Large] (output) at (\boxR+1.1,0) {$\hat{\mathbf{u}}$};

            \draw[->] (nodeh.east) -- node[above] {$\mathbf{y}_h=\mathbf{u}+\mathbf{n}_h$} (\boxL,\nodeY);
            \draw[->] (\boxL,\nodeY) -- (decision.155);

            \draw[->] (nodea.east) -- node[above] {$\mathbf{y}_a=\mathbf{u}+\mathbf{n}_a$} (\boxL,-\nodeY);
            \draw[->] (\boxL,-\nodeY) -- (decision.205);

            \draw[->] (decision.east) -- node[above] {Yes} (estimation.west);
            \draw[->] (decision.south) -- node[right] {No} (reject.north);

            \draw[->] (estimation.east) -- (output.west);

        \end{tikzpicture}%
    }
    \vspace{-10pt}
   \caption{Two-node game-of-coding model. The DC first commits to $\eta$. The honest node reports $\mathbf{y}_h=\mathbf{u}+\mathbf{n}_h$, where $|\mathbf{n}_h| \leq \Delta$, while the adversarial node reports $\mathbf{y}_a=\mathbf{u}+\mathbf{n}_a$, where the distribution of $\mathbf{n}_a$ is chosen arbitrarily by the adversary in response to $\eta$. The DC accepts the reports if $|\mathbf{y}_h-\mathbf{y}_a|\leq \eta\Delta$ and then forms an estimate $\hat{u}$; otherwise, it rejects them.}
    \label{fig:two_node_model}
\end{figure}
The DC receives $\underline{\mathbf{y}}\triangleq\{\mathbf{y}_1,\mathbf{y}_2\}$ and evaluates it as follows:
\begin{itemize}
    \item First, the DC accepts $\underline{\mathbf{y}}$ if and only if $|\mathbf{y}_1-\mathbf{y}_2|\leq \eta\Delta$, for $\eta\geq 2$. We denote by $\mathcal{A}_{\eta}$ the event that the inputs are accepted, and define $\mathsf{PA}(g(.),\eta)\triangleq \Pr(\mathcal{A}_{\eta})$.
    
    \item Second, if the inputs are accepted, the DC outputs $\mathsf{est}(\underline{\mathbf{y}})$ as its estimate of $\mathbf{u}$, where $\mathsf{est}:\mathbb{R}^2\to\mathbb{R}$. The estimation cost is defined as $\mathsf{MSE}\left(\mathsf{est}(.,.),g(.),\eta\right)\triangleq \mathbb{E}[(\mathbf{u}-\mathsf{est}(\underline{\mathbf{y}}))^2\mid \mathcal{A}_{\eta}]$. The DC chooses the estimator $\mathsf{est}^{*}_{\eta,g}\triangleq \underset{\mathsf{est}:\mathbb{R}^2\to\mathbb{R}}{\arg\min}~\mathsf{MSE}\left(\mathsf{est}(.,.),g(.),\eta\right)$. We define 
    \begin{align}
        \mathsf{MMSE}(g(.),\eta)
        \triangleq
        \mathsf{MSE}\left(\mathsf{est}^{*}_{\eta,g}(.,.),g(.),\eta\right).
    \end{align}
\end{itemize}

The interaction between the DC and the adversary is modeled as a two-player Stackelberg game~\cite{von2010market}. The DC utility and the adversary utility are respectively defined as $\mathsf{U}_{\mathsf{DC}}(g(.),\eta)\triangleq \mathsf{Q}_{\mathsf{DC}}\left(\mathsf{MMSE},\mathsf{PA}\right)$ and $\mathsf{U}_{\mathsf{AD}}(g(.),\eta)\triangleq \mathsf{Q}_{\mathsf{AD}}\left(\mathsf{MMSE},\mathsf{PA}\right)$. The function $\mathsf{Q}_{\mathsf{DC}}:\mathbb{R}^2\to\mathbb{R}$ is non-increasing with respect to $\mathsf{MMSE}$ and non-decreasing with respect to $\mathsf{PA}$. Conversely, $\mathsf{Q}_{\mathsf{AD}}:\mathbb{R}^2\to\mathbb{R}$ is strictly increasing with respect to both arguments. We assume that the adversary knows both utility functions, whereas the DC  does not know $\mathsf{Q}_{\mathsf{AD}}$.

We assume that the DC chooses  $\eta \in \Lambda_{\mathsf{DC}}\triangleq [\eta_{\min},\eta_{\max}]$, where $2\leq \eta_{\min}<\eta_{\max}$, and that the adversary selects its strategy from $\Lambda_{\mathsf{AD}}\triangleq \left\{g(.)~\middle|~g(.):\mathbb{R}\to\mathbb{R}\text{ is a valid PDF}\right\}$.
For each $\Lambda_{\mathsf{DC}}$ to which the DC commits, the adversary observes $\eta$ and chooses a best response $\mathcal{B}^{\eta}
    \triangleq
    \arg\max
    ~\mathsf{U}_{\mathsf{AD}}(g(.),\eta)$.
Elements of $\mathcal{B}^{\eta}$ provide the same utility to the adversary, but they may induce different utilities for the DC. We define
\begin{align}
    \Bar{\mathcal{B}}^{\eta}
    \triangleq
    \underset{g(.)\in\mathcal{B}^{\eta}}{\arg\min}
    ~\mathsf{U}_{\mathsf{DC}}(g(.),\eta).
\end{align}
For every $\eta \in \Lambda_{\mathsf{DC}}$, fix an arbitrary $g^*_{\eta}(.)\in\Bar{\mathcal{B}}^{\eta}$ and define $\alpha(\eta)\triangleq \mathsf{PA}(g^*_{\eta}(.),\eta)$ and $m(\eta)\triangleq \mathsf{MMSE}(g^*_{\eta}(.),\eta)$.
Thus, the induced utility of the DC when it commits to $\eta$ is
\begin{align}
    \mathsf{U}(\eta)
    \triangleq
    \mathsf{U}_{\mathsf{DC}}(g^*_{\eta}(.),\eta)
    =
    \mathsf{Q}_{\mathsf{DC}}
    \left(
    m(\eta),
    \alpha(\eta)
    \right).
    \label{def:induced_utility}
\end{align}
We define the optimal utility as $\mathsf{U}^*_{\Lambda_{\mathsf{DC}}}
    \triangleq
    \sup_{\eta \in \Lambda_{\mathsf{DC}}}\mathsf{U}(\eta)$.

We now define the multi-round interaction. At the beginning of round $t$, the DC has access to the history $\mathcal{H}_{t-1}=\{\eta_s,\underline{\mathbf{y}}_s,\mathbbm{1}\{\mathcal{A}_{\eta_s}\}\}_{s=1}^{t-1}$. A policy $\pi$ of the DC is a sequence of maps $\pi_t:\mathcal{H}_{t-1}\to\Lambda_{\mathsf{DC}}$, and the DC commits to $\eta_t=\pi_t(\mathcal{H}_{t-1})$. After observing $\eta_t$, the adversary myopically chooses $g_t(.)\in\Bar{\mathcal{B}}^{\eta_t}$. The term \emph{myopic} means that the adversary maximizes its immediate utility for the current round $t$ with respect to the value of $\eta_t$. The DC observes the reports $\underline{\mathbf{y}}_t$ and the acceptance indicator $Z_t\triangleq \mathbf{1}\{\mathcal{A}_{\eta_t}\}$. Conditioned on $\mathcal{H}_{t-1}$ and $\eta_t$, the acceptance indicator satisfies
\begin{align}
    \mathbb{E}\left[Z_t\mid \mathcal{H}_{t-1},\eta_t\right]
    =
    \alpha(\eta_t).
    \label{eq:acceptance_feedback}
\end{align}

The objective is to design a policy $\pi=\{\pi_t\}_{t\geq 1}$ whose performance is good over the entire trajectory. We measure the performance of $\pi$, by the cumulative regret with respect to the best fixed acceptance parameter in $\Lambda_{\mathsf{DC}}$:
\begin{align}
    \mathsf{Reg}_T(\pi)
    =
    \mathbb{E}_{\pi}
   [
    \sum_{t=1}^{T}
    (
    \mathsf{U}^*_{\Lambda_{\mathsf{DC}}}-\mathsf{U}(\eta_t)
    )
    ].\label{def:regret}
\end{align}


\section{Preliminaries}
\label{sec:preliminaries}

For any $\eta\in\Lambda_{\mathsf{DC}}$ and $\alpha\in[0,1]$, we define
\begin{align*}
    c_{\eta}(\alpha)
    \triangleq
    \max_{g(.)\in\Lambda_{\mathsf{AD}}}
    \mathsf{MMSE}(g(.),\eta)
    \quad
    \text{s.t.}\quad
    \mathsf{PA}(g(.),\eta)\geq \alpha .
\end{align*}
Note that the function $c_{\eta}(.)$ does not depend on $\mathsf{Q}_{\mathsf{DC}}$ or $\mathsf{Q}_{\mathsf{AD}}$. A key characterization from the game-of-coding framework~\cite{GoCJournal, nodehi2026gamevector, GoDSybil, nodehi2026game, nodehi2025unknown} states that, for any $g^*_{\eta}(.)\in\Bar{\mathcal{B}}^{\eta}$, we have $\mathsf{MMSE}(g^*_{\eta}(.),\eta)=c_{\eta}\left(\mathsf{PA}(g^*_{\eta}(.),\eta)\right)$. This implies that $m(\eta)=c_{\eta}(\alpha(\eta))$.
Define  $F_{\eta}(\alpha)\triangleq \mathsf{Q}_{\mathsf{DC}}\left(c_{\eta}(\alpha),\alpha\right)$. Based on  \eqref{def:induced_utility}, we have $\mathsf{U}(\eta)=\mathsf{Q}_{\mathsf{DC}}\left(c_{\eta}(\alpha(\eta)),\alpha(\eta)\right)=F_{\eta}(\alpha(\eta))$.

This identity is central to our algorithm. The DC does not know $\mathsf{Q}_{\mathsf{AD}}$, does not know the adversary's selected response $g^*_{\eta}(.)$, and does not observe samples with mean $\mathsf{U}(\eta)$. Instead, when the DC commits to $\eta_t$, it observes the binary acceptance indicator $Z_t$, whose conditional expectation is $\alpha(\eta)$ by \eqref{eq:acceptance_feedback}. Repeated plays of the same $\eta$ therefore allow the DC to estimate $\alpha(\eta)$. The function $F_{\eta}(.)$ then converts this estimate of $\alpha(\eta)$ into an estimate of the induced utility $\mathsf{U}(\eta)$.

\section{Main Results}
\label{sec:red-main-result}

In this section, we present the main results. Let 
\begin{align}
    \Gamma(\eta)
    \triangleq
    \mathsf{U}^*_{\Lambda_{\mathsf{DC}}}
    -
    \mathsf{U}(\eta),
    \label{eq:red-gap-definition}
\end{align}
for all $\eta\in\Lambda_{\mathsf{DC}}$. We impose the following assumptions, which are discussed in Appendix~\ref{app:assumptions-and-zooming}.

\noindent\textbf{A1.}
For all $\eta\in\Lambda_{\mathsf{DC}}$, we assume that $0\leq \mathsf{U}(\eta)\leq 1$.

\noindent\textbf{A2.}
We assume that there exists a metric $\mathcal{D}$ such that $\left|\mathsf{U}(\eta)-\mathsf{U}(\eta')\right|\leq \mathcal{D}(\eta,\eta')$ for all $\eta,\eta'\in\Lambda_{\mathsf{DC}}$. The diameter of a set $\mathcal{S}$ under $\mathcal{D}$ is defined as $\operatorname{diam}\left(\mathcal{S},\mathcal{D}\right)\triangleq \sup_{\eta,\eta'\in\mathcal{S}}\mathcal{D}(\eta,\eta')$. We assume that $\operatorname{diam}\left(\Lambda_{\mathsf{DC}},\mathcal{D}\right)\leq 1$.

\noindent\textbf{A3.}
For all $\eta\in\Lambda_{\mathsf{DC}}$ and all $\alpha,\alpha'\in[0,1]$, we assume that $\left|F_{\eta}(\alpha)-F_{\eta}(\alpha')\right|\leq \ell |\alpha-\alpha'|$. We define $\bar{\ell}\triangleq \max\{1,\ell\}$.

\noindent\textbf{A4.}
 For every
$r\in(0,1]$, define the gap annulus
\begin{align}
    \Lambda_r
    \triangleq
  \{
        \eta\in\Lambda_{\mathsf{DC}}
        :
        \frac{r}{2}
        <
        \Gamma(\eta)
        \leq
        r
   \}.
    \label{eq:red-ksu-annulus}
\end{align}
We assume that the game-of-coding  instance $\mathcal{I}_{\mathsf{GoC}}
    \triangleq
    \left(
        \Lambda_{\mathsf{DC}},
        \mathcal{D},
        \mathsf{U}
    \right)$ has
$(C_z,d_z)$-zooming dimension, meaning that for every
$r\in(0,1]$, the set $\Lambda_r$ can be covered by at most $C_z r^{-d_z}$
sets of $\mathcal{D}$-diameter at most $\frac{r}{8}$.

\begin{algorithm}[t]
\caption{Game-of-Coding Zooming Algorithm}
\label{alg:red-goc-zooming}
\begin{algorithmic}[1]
\Require Metric $\mathcal{D}$, Lipschitz constant $\ell$, maps $\{F_{\eta}(.)\}_{\eta\in\Lambda_{\mathsf{DC}}}$
\Ensure Commitments $\eta_1,\eta_2,\ldots$
\State $\bar{\ell}\gets\max\{1,\ell\}$, $t\gets0$
\For{phases $i=1,2,\ldots$}
    \State $\mathcal{S}_i\gets\emptyset$
    \For{local rounds $s=1,\ldots,2^i$}
        \State $t\gets t+1$
        \For{each $v\in\mathcal{S}_i$}
            \State $\widehat{\alpha}_i(v)\gets N_i^{\mathsf{acc}}(v)/N_i(v)$ if $N_i(v)>0$, and $0$ otherwise
            \State $\widehat{\mathsf{U}}_i(v)\gets \Pi_{[0,1]}(F_v(\widehat{\alpha}_i(v)))$
            \State $\rho_i(v)\gets \bar{\ell}\sqrt{8i/(2+N_i(v))}$
        \EndFor
        \If{$\Lambda_{\mathsf{DC}}\nsubseteq \displaystyle\bigcup_{v\in\mathcal{S}_i}\mathcal{B}_{\mathcal{D}}(v,\rho_i(v))$}
            \State Add an uncovered $\eta_n$ to $\mathcal{S}_i$ with $N_i(\eta_n)=N_i^{\mathsf{acc}}(\eta_n)=0$
        \EndIf
        \State Choose $v_i\in\arg\max_{v\in\mathcal{S}_i}\left\{\widehat{\mathsf{U}}_i(v)+2\rho_i(v)\right\}$
        \State Commit to $\eta_t=v_i$ and observe $Z_t=\mathbf{1}\{\mathcal{A}_{\eta_t}\}$
        \State $N_i(v_i)\gets N_i(v_i)+1$, \quad $N_i^{\mathsf{acc}}(v_i)\gets N_i^{\mathsf{acc}}(v_i)+Z_t$
    \EndFor
\EndFor
\end{algorithmic}
\end{algorithm}

For $x\in\mathbb{R}$, let  $\Pi_{[0,1]}(x)\triangleq \min\{1,\max\{0,x\}\}$. Also, for $v\in\Lambda_{\mathsf{DC}}$ and $r>0$, let  $\mathcal{B}_{\mathcal{D}}(v,r)\triangleq \{\eta\in\Lambda_{\mathsf{DC}}:\mathcal{D}(\eta,v)\leq r\}$.

We now present Algorithm~\ref{alg:red-goc-zooming}, which is based on the phase-based zooming algorithm of~\cite{kleinberg2008multi}. The algorithm runs in phases, where phase $i$ has length $2^i$ rounds. At the beginning of each phase, the active set is reset. We denote by $\mathcal{S}_i$ the set of active thresholds in phase $i$. For each active threshold $v\in\mathcal{S}_i$, let $N_i(v)$ be the number of times $v$ has been played so far in phase $i$, and let $N_i^{\mathsf{acc}}(v)$ be the number of accepted plays. Thus, if $v$ has been played in rounds $t_1,\ldots,t_{N_i(v)}$ of phase $i$, then $N_i^{\mathsf{acc}}(v)=\sum_{j=1}^{N_i(v)} Z_{t_j}$. When $N_i(v)>0$, the empirical acceptance probability is $\widehat{\alpha}_i(v)=N_i^{\mathsf{acc}}(v)/N_i(v)$, and the  utility estimate is $\widehat{\mathsf{U}}_i(v)=\Pi_{[0,1]}\left(F_v(\widehat{\alpha}_i(v))\right)$.

The confidence radius assigned to $v$ in phase $i$ is $\rho_i(v)=\bar{\ell}\sqrt{8i/(2+N_i(v))}$. This radius captures the statistical uncertainty in $\widehat{\mathsf{U}}_i(v)$ and also determines the region of the decision space represented by $v$, namely $\mathcal{B}_{\mathcal{D}}(v,\rho_i(v))$. The optimistic index of $v$ is $I_i(v)=\widehat{\mathsf{U}}_i(v)+2\rho_i(v)$, where the first term estimates the utility and the second term provides an optimism bonus. The algorithm plays active thresholds with large optimistic index and activates a new threshold only when the current active balls do not cover the entire decision space.

The intuition behind Algorithm~\ref{alg:red-goc-zooming} is simple. As an active threshold $v$ is sampled more often, $N_i(v)$ increases and $\rho_i(v)$ decreases, giving a more accurate estimate of its utility. If some region of $\Lambda_{\mathsf{DC}}$ is not covered by the current active balls, the algorithm activates a new threshold in that region. The algorithm plays the active threshold with the largest optimistic index $I_i(v)$. In this way, the algorithm adaptively concentrates its search around thresholds that may be near-optimal.

\begin{theorem}
\label{thm:main-zooming}
\label{thm:red-main-zooming}
Under the  Assumptions A1--A4, let
$\pi_{\mathsf{Z}}$ be the policy induced by
Algorithm~\ref{alg:red-goc-zooming}. Then, for every $T\geq2$,
\begin{align}
    \mathsf{Reg}_T(\pi_{\mathsf{Z}})
    =
    O
    \left(
        \left(
            C_z\bar{\ell}^{2}\ln^{2}T
        \right)^{\frac{1}{d_z+2}}
        T^{\frac{d_z+1}{d_z+2}}
    \right).
    \label{eq:red-main-regret-rate}
\end{align}
\end{theorem}
Theorem~\ref{thm:red-main-zooming} implies that since $(d_z+1)/(d_z+2)<1$,
we have $\mathsf{Reg}_T(\pi_{\mathsf{Z}})/{T}\to
    0$, as $T\to\infty$.
Thus, although the DC observes only binary acceptance feedback, its
average loss relative to the best fixed threshold in hindsight vanishes.
Moreover, 
The regret exponent in Theorem~\ref{thm:red-main-zooming} is the same
zooming-type exponent obtained for metric-space bandits with zooming
dimension $d_z$~\cite{kleinberg2008multi}.

\section{Proof of Theorem~\ref{thm:red-main-zooming}}
\label{sec:red-proof-main-result}

In this section, we prove Theorem~\ref{thm:red-main-zooming}. Note that the DC does not observe samples with mean $\mathsf{U}(\eta)$. Instead, by \eqref{eq:acceptance_feedback}, it observes binary acceptance indicators with mean $\alpha(\eta)$. The next lemma shows that these observations still provide valid confidence intervals for the  utility $\mathsf{U}(\eta)=F_{\eta}(\alpha(\eta))$.

\begin{lemma}
\label{lem:proof-clean-phase}
For every phase $i$, every active threshold $v\in\mathcal{S}_i$ and every time in phase $i$, with probability at least $1-8^{-i}$, we have $|
        \widehat{\mathsf{U}}_i(v)
        -
        \mathsf{U}(v)
    |
    \leq
    \rho_i(v)$.
\end{lemma}

\begin{proof}
Fix a phase $i$, and let $V_{i,1},\ldots,V_{i,J_i}$ be the thresholds activated during this phase, listed in their activation order. Since phase $i$ has length $2^i$ and at most one new threshold can be activated in each round, we have $J_i\leq 2^i$.

For an activated threshold $V_{i,j}$, let $Z_{i,j,1},\ldots,Z_{i,j,s}$ denote the first $s$ acceptance indicators observed from plays of $V_{i,j}$ in phase $i$, and define $\widehat{\alpha}_{i,j,s}\triangleq s^{-1}\sum_{\tau=1}^{s}Z_{i,j,\tau}$. By \eqref{eq:acceptance_feedback}, each acceptance indicator observed after committing to $V_{i,j}$ has conditional mean $\alpha(V_{i,j})$. Therefore, Hoeffding's inequality~\cite{hoeffding1994probability} gives, for every fixed $j$ and $s\geq1$,
\begin{align*}
    \Pr(
        |
            \widehat{\alpha}_{i,j,s}
            -
            \alpha(V_{i,j})
        |
        >
        \sqrt{8i/(2+s)}
    )
    \leq
    2\exp(
        -\frac{16is}{2+s}
    ).
\end{align*}
Since $s/(2+s)\geq 1/3$ for $s\geq1$, the right-hand side is at most $2e^{-16i/3}$. There are at most $2^i$ activated thresholds and at most $2^i$ possible sample sizes for each threshold. Hence, by the union bound, the probability that any empirical acceptance estimate in phase $i$ violates the above bound is at most $2^i\cdot 2^i\cdot 2e^{-16i/3}\leq 8^{-i}$. Thus, with probability at least $1-8^{-i}$, for all activated thresholds $V_{i,j}$ and all sample sizes $s\leq 2^i$,
\begin{align}
    \left|
        \widehat{\alpha}_{i,j,s}
        -
        \alpha(V_{i,j})
    \right|
    \leq
    \sqrt{8i/(2+s)}.
    \label{eq:proof-clean-alpha}
\end{align}

Now suppose this event holds, and let $v\in\mathcal{S}_i$ be active. If $N_i(v)=s>0$, then $v=V_{i,j}$ for some $j$ and $\widehat{\alpha}_i(v)=\widehat{\alpha}_{i,j,s}$. By definition, $\widehat{\mathsf{U}}_i(v)=\Pi_{[0,1]}(F_v(\widehat{\alpha}_i(v)))$. Since $\mathsf{U}(v)\in[0,1]$ by Assumption A1 and projection onto $[0,1]$ is non-expansive,
\begin{align}
    |
        \widehat{\mathsf{U}}_i(v)
        -
        \mathsf{U}(v)
    |
    \leq
    \left|
        F_v(\widehat{\alpha}_i(v))
        -
        F_v(\alpha(v))
    \right|.
\end{align}
Using Assumption A3 and \eqref{eq:proof-clean-alpha}, we obtain $|
        \widehat{\mathsf{U}}_i(v)
        -
        \mathsf{U}(v)
    |
    \leq
    \ell
    \left|
        \widehat{\alpha}_i(v)
        -
        \alpha(v)
    \right|
    \leq
    \bar{\ell}
    \sqrt{
        8i/(2+N_i(v))
    }
    =
    \rho_i(v)$.

If $N_i(v)=0$, then Algorithm~\ref{alg:red-goc-zooming} sets $\widehat{\mathsf{U}}_i(v)=0$. By Assumption A1, $|\widehat{\mathsf{U}}_i(v)-\mathsf{U}(v)|\leq 1$. Also, $\rho_i(v)=\bar{\ell}\sqrt{8i/2}=2\bar{\ell}\sqrt{i}\geq1$, since $i\geq1$ and $\bar{\ell}\geq1$. Hence Lemma \ref{lem:proof-clean-phase} also holds when $N_i(v)=0$. This proves the lemma.
\end{proof}

To complete the proof, we invoke \cite[Theorem~4.2]{kleinberg2008multi}. The theorem applies to the phase-based zooming algorithm of \cite[Algorithm~2.3]{kleinberg2008multi}, which maintains a finite active set of arms in each phase. At each round, the algorithm checks whether the confidence balls around the active arms cover the whole metric space. If not, it activates an uncovered arm. It then assigns each active arm $v$ the optimistic index $\widehat{\mu}_t(v)+2\widehat r_t(v)$ and plays an active arm with the largest index, where $\widehat{\mu}_t(v)$ is an estimator of the payoff and $\widehat r_t(v)$ is its confidence radius.

In the notation of \cite[Theorem~4.2]{kleinberg2008multi}, the bandit instance is $(L,X,\mu)$, where $X$ is the arm set, $L$ is the metric, and $\mu$ is the payoff function. The theorem assumes that $\mu:X\to[0,1]$, $\operatorname{diam}(X,L)\leq1$, and $|\mu(x)-\mu(y)|\leq L(x,y)$ for all $x,y\in X$. It also defines $\mu^*\triangleq \sup_{x\in X}\mu(x)$ and $\Delta(x)\triangleq \mu^*-\mu(x)$.

The theorem requires the instance to have zooming dimension $d$ with constant $c$, meaning that, for every $r\in(0,1]$, the set $X_r\triangleq \{x\in X:r/2<\Delta(x)\leq r\}$ can be covered by at most $cr^{-d}$ sets of $L$-diameter at most $r/8$. It also requires clean phases, namely $|\widehat{\mu}_t(v)-\mu(v)|\leq \widehat r_t(v)$ for every active arm $v$ and every time $t$ in phase $i$, with probability at least $1-8^{-i}$. Finally, the radius must be smooth, meaning $(3/4)\widehat r_t(v)\leq \widehat r_{t+1}(v)\leq \widehat r_t(v)$, and $(q,c_0)$-good, meaning that $\Delta(v)\leq 4\widehat r_t(v)$ implies $n_t(v)\leq c_0 i\,\Delta(v)^{-q}$, where $n_t(v)$ is the number of times $v$ has been played earlier in  phase $i$.

Under these conditions, if $d+q>1$, \cite[Theorem~4.2]{kleinberg2008multi} implies that the regret of the  zooming policy satisfies
\begin{align}
    \mathsf{Reg}_T(\pi)
    \leq
    C_0
    \left(
        c c_0 \ln^2(eT)
    \right)^{\frac{1}{d+q}}
    T^{1-\frac{1}{d+q}},
    \label{eq:proof-ext-bound}
\end{align}
for a universal constant $C_0>0$.

We now verify these conditions for Algorithm~\ref{alg:red-goc-zooming}. We identify $X=\Lambda_{\mathsf{DC}}$, $L=\mathcal{D}$, and $\mu=\mathsf{U}$. By Assumption A1, $\mu(\eta)=\mathsf{U}(\eta)\in[0,1]$ for all $\eta\in\Lambda_{\mathsf{DC}}$. By Assumption A2, $\operatorname{diam}(\Lambda_{\mathsf{DC}},\mathcal{D})\leq1$ and $|\mu(\eta)-\mu(\eta')|=|\mathsf{U}(\eta)-\mathsf{U}(\eta')|\leq\mathcal{D}(\eta,\eta')$ for all $\eta,\eta'\in\Lambda_{\mathsf{DC}}$. Thus, the boundedness, diameter, and Lipschitz requirements are satisfied.
The optimal value in the induced bandit instance is $\mu^*=\sup_{\eta\in\Lambda_{\mathsf{DC}}}\mathsf{U}(\eta)=\mathsf{U}^*_{\Lambda_{\mathsf{DC}}}$. Hence the bandit gap is $\Delta(\eta)=\mu^*-\mu(\eta)=\Gamma(\eta)$, so the set $X_r$ is exactly the set $\Lambda_r$ in \eqref{eq:red-ksu-annulus}. By Assumption A4, the zooming-dimension condition holds with $c=C_z$ and $d=d_z$.

It remains to verify the algorithmic conditions. Under the above identification, Algorithm~\ref{alg:red-goc-zooming} has the same phase structure, activation rule, covering rule, and optimistic-index rule as the phase-based zooming algorithm in \cite[Algorithm~2.3]{kleinberg2008multi}. Its estimator and radius are $\widehat{\mu}_t(v)=\widehat{\mathsf{U}}_i(v)$ and $\widehat r_t(v)=\rho_i(v)$, where $i$ is the current phase. Lemma~\ref{lem:proof-clean-phase} proves the clean-phase condition. It remains only to prove the smoothness and goodness conditions, where we prove in the following lemma.

\begin{lemma}
\label{lem:proof-radius-conditions}
for all $i$, the radius $\rho_i(v)$ satisfies the smoothness condition. Moreover, it is $(2,128\bar{\ell}^2)$-good.
\end{lemma}

\begin{proof}
Let $n=N_i(v)$ be the number of times $v$ has been played so far in phase $i$. If $v$ is not played in the next round, then the radius does not change. If $v$ is played, then $n$ becomes $n+1$, and based on $\rho_i(v) =\bar{\ell}
    \sqrt{
        8i/(2+N_i(v))}$, we have
\begin{align}
    \frac{\rho_i(n+1)}{\rho_i(n)}
    =
    \sqrt{
        \frac{2+n}{3+n}
    }
    \geq
    \sqrt{\frac{2}{3}}
    >
    \frac{3}{4}.
\end{align}
Therefore, $(3/4)\rho_i(v)\leq \rho_i^+(v)\leq \rho_i(v)$, where $\rho_i^+(v)$ is the radius after the next round. This proves smoothness.

We now prove goodness. Suppose $\Gamma(v)\leq 4\rho_i(v)$. This gives $\Gamma(v)\leq 4\bar{\ell}\sqrt{8i/(2+N_i(v))}$. Squaring both sides and rearranging yields $2+N_i(v)\leq 128\bar{\ell}^2 i\,\Gamma(v)^{-2}$, and therefore $N_i(v)\leq 128\bar{\ell}^2 i\,\Gamma(v)^{-2}$. Since $\Gamma(v)=\Delta(v)$, this is the goodness condition with $q=2$ and $c_0=128\bar{\ell}^2$.
\end{proof}

We now apply \cite[Theorem~4.2]{kleinberg2008multi} with $c=C_z$, $d=d_z$, $q=2$, and $c_0=128\bar{\ell}^2$. Since $d_z+2>1$, \eqref{eq:proof-ext-bound} gives
\begin{align}
    \mathsf{Reg}_T(\pi_{\mathsf Z})
    \leq
    C_0
    \left(
        128C_z\bar{\ell}^2
        \ln^2(eT)
    \right)^{\frac{1}{d_z+2}}
    T^{\frac{d_z+1}{d_z+2}}.
    \label{eq:proof-applied-bound}
\end{align}
This completes the proof of Theorem~\ref{thm:red-main-zooming}.

\section{Numerical Experiments}
\label{sec:experiments}

In this section, we compare Algorithm~\ref{alg:red-goc-zooming} with the explore-then-commit baseline of~\cite{nodehi2025unknown}. In the unknown-adversary formulation of~\cite{nodehi2025unknown}, the DC first explores for a fixed number of rounds and is evaluated only through the final selected threshold $\hat{\eta}$, with guarantee $\Pr\left(\mathsf{U}^*_{\Lambda_{\mathsf{DC}}}-\mathsf{U}(\hat{\eta})>\lambda\right)<\delta$. In contrast, our formulation evaluates the entire learning trajectory through the cumulative regret in \eqref{def:regret}. 

We consider the two-node setting from Section~\ref{sec:problem-formulation}. The honest noise is $\mathbf{n}_h\sim\mathrm{Unif}[-\Delta,\Delta]$ with $\Delta=1$, and the threshold space is $\Lambda_{\mathsf{DC}}=[2,30]$. For each fixed $\eta$, the adversary chooses the induced acceptance probability as $\alpha(\eta)\in \arg\max_{0<\alpha\leq1}\mathsf{Q}_{\mathsf{AD}}\left(c_{\eta}(\alpha),\alpha\right)$. The corresponding DC utility is $\mathsf{U}(\eta)=\mathsf{Q}_{\mathsf{DC}}\left(c_{\eta}(\alpha(\eta)),\alpha(\eta)\right)$.

In the experiment, we use $\mathsf{Q}_{\mathsf{AD}}(\mathsf{MMSE},\mathsf{PA})=\ln(\mathsf{MMSE})+0.2\ln(\mathsf{PA})$ and $\mathsf{Q}_{\mathsf{DC}}(\mathsf{MMSE},\mathsf{PA})=-\mathsf{MMSE}+200\mathsf{PA}$. Using the procedure described in \cite[Algorithm~1]{GoCJournal}, the best threshold over $[2,30]$ is $\eta^*=12.7189$, with $\mathsf{U}^*_{2,30}=44.7935$. Fig.~\ref{fig:exp-true-utility} shows the induced DC utility curve, where the dashed vertical line marks $\eta^*$.

\begin{figure}[t]
    \centering
    \includegraphics[width=0.75\linewidth]{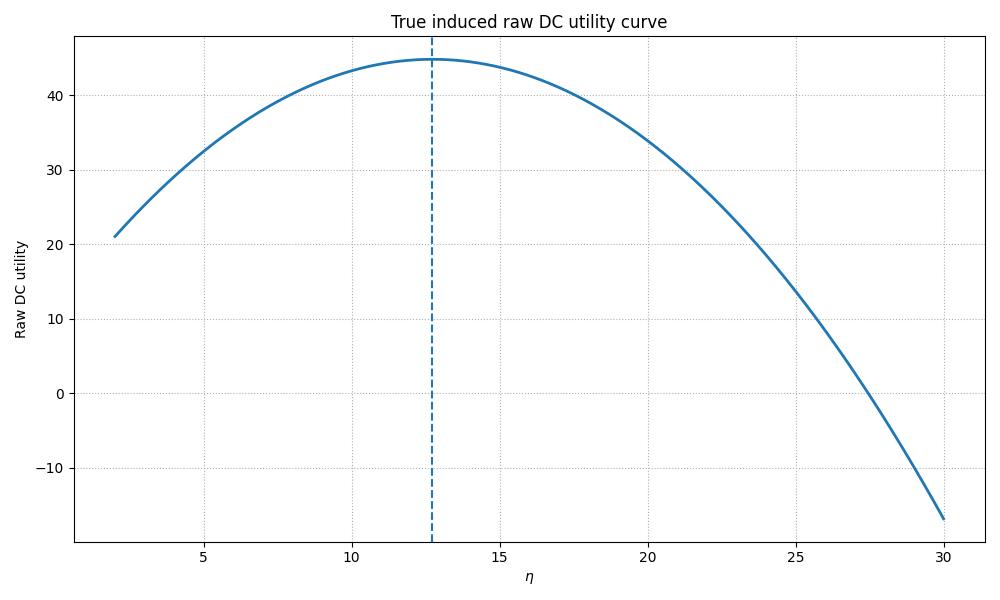}
    \caption{Induced DC utility as a function of the acceptance threshold $\eta$. The dashed vertical line marks the optimal threshold $\eta^*=12.7189$.}
    \label{fig:exp-true-utility}
\end{figure}

For the explore-then-commit baseline, we set $\lambda=0.5$ and $\delta=0.05$. Following~\cite{nodehi2025unknown}, this gives $n=13$ and $k=3280$. Hence, the baseline explores $n+1=14$ grid points and spends $(n+1)k=45920$ rounds in the exploration phase.

We run both algorithms for $T=100000$ rounds. In each round, we sample $\mathbf{u}\sim\mathrm{Unif}[-M,M]$ and $\mathbf{n}_a$ is sampled  following the procedure described in \cite[Algorithm~2]{GoCJournal}. 

The baseline spends $45920$ rounds exploring and then commits for the remaining $54080$ rounds. It selects $\hat{\eta}_{\mathrm{old}}=10.6154$, achieving final cumulative  regret $14058.9731$. Algorithm~\ref{alg:red-goc-zooming} activates only $11$ thresholds over the trajectory.  It achieves final cumulative  regret $8852.6987$. Fig.~\ref{fig:exp-regret} compares the cumulative  regret of the two algorithms. 

\begin{figure}[t]
    \centering
    \includegraphics[width=0.75\linewidth]{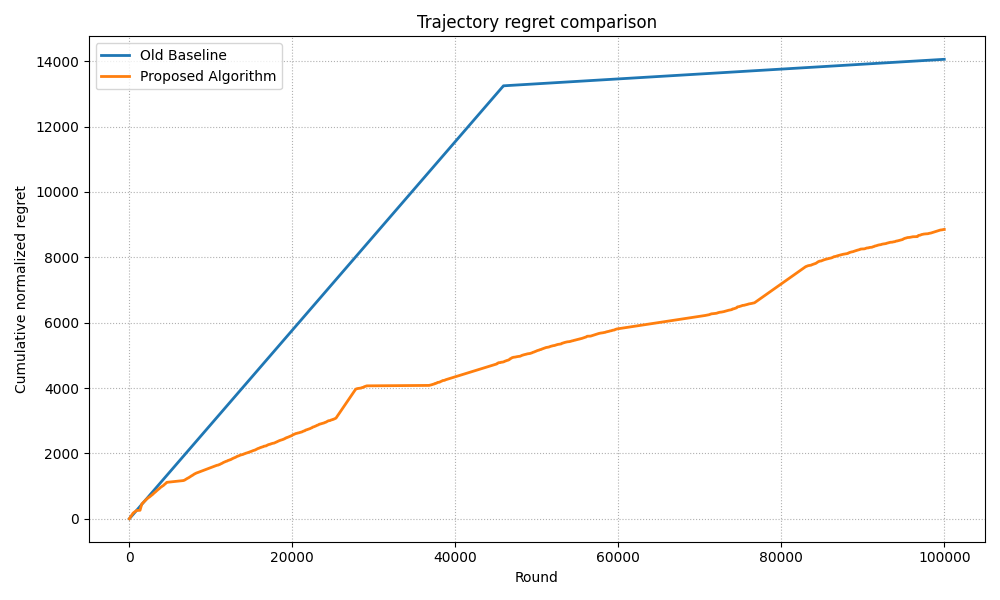}
    \caption{Cumulative  regret over $T=100000$ rounds. The proposed  algorithm achieves lower trajectory regret than the explore-then-commit baseline.}
    \label{fig:exp-regret}
\end{figure}

\section{Acknowledgment}
This work has been partially supported by the National Science Foundation under Grant CCF-2348638.

\appendices

\section{Discussion of Assumptions and Zooming Dimension}
\label{app:assumptions-and-zooming}

This appendix discusses the assumptions used in
Theorem~\ref{thm:red-main-zooming}. 

\subsection{Utility Normalization}
\label{subsec:app-utility-normalization}

Theorem~\ref{thm:red-main-zooming} assumes
\begin{align}
    0\leq \mathsf{U}(\eta)\leq 1,
    \qquad
    \forall \eta\in\Lambda_{\mathsf{DC}}.
    \label{eq:app-normalized-U}
\end{align}
This is a standard normalization. Suppose instead that the induced
utility is bounded as
\begin{align}
    U_{\min}
    \leq
    \mathsf{U}(\eta)
    \leq
    U_{\max},
    \qquad
    \forall \eta\in\Lambda_{\mathsf{DC}},
    \label{eq:app-utility-range}
\end{align}
where
\begin{align}
    R_U
    \triangleq
    U_{\max}-U_{\min}.
    \label{eq:app-utility-range-width}
\end{align}
If $R_U=0$, then all thresholds have the same utility and the
regret is identically zero. Otherwise, define the normalized utility
\begin{align}
    \widetilde{\mathsf{U}}(\eta)
    \triangleq
    \frac{
        \mathsf{U}(\eta)-U_{\min}
    }{
        R_U
    }.
    \label{eq:app-normalized-utility}
\end{align}
Then
\begin{align}
    0
    \leq
    \widetilde{\mathsf{U}}(\eta)
    \leq
    1.
    \label{eq:app-normalized-utility-range}
\end{align}
The corresponding optimal value is
\begin{align}
    \widetilde{\mathsf{U}}^*
    =
    \frac{
        \mathsf{U}^*_{\Lambda_{\mathsf{DC}}}-U_{\min}
    }{
        R_U
    }.
    \label{eq:app-normalized-optimum}
\end{align}
Therefore, the normalized gap satisfies
\begin{align}
    \widetilde{\Gamma}(\eta)
    &\triangleq
    \widetilde{\mathsf{U}}^*
    -
    \widetilde{\mathsf{U}}(\eta)
    \nonumber\\
    &=
    \frac{
        \mathsf{U}^*_{\Lambda_{\mathsf{DC}}}
        -
        \mathsf{U}(\eta)
    }{
        R_U
    }
    =
    \frac{\Gamma(\eta)}{R_U}.
    \label{eq:app-normalized-gap}
\end{align}
Thus, for any policy $\pi$,
\begin{align}
    \widetilde{\mathsf{Reg}}_T(\pi)
    =
    \frac{1}{R_U}
    \mathsf{Reg}_T(\pi).
    \label{eq:app-normalized-regret}
\end{align}
A regret bound for the normalized utility therefore gives the
corresponding original-scale regret bound after multiplication by
$R_U$. Hence, assuming \eqref{eq:app-normalized-U} does not change
the substance of the result.

\subsection{Metric Normalization}
\label{subsec:app-metric-normalization}

Theorem~\ref{thm:red-main-zooming} assumes that the metric space has
diameter at most one:
\begin{align}
    \operatorname{diam}
    \left(
        \Lambda_{\mathsf{DC}},
        \mathcal{D}
    \right)
    \leq
    1.
    \label{eq:app-diameter-normalized}
\end{align}
This is only a normalization.

Suppose first that we have a metric $\mathcal{D}_0$ such that
\begin{align}
    |\mathsf{U}(\eta)-\mathsf{U}(\eta')|
    \leq
    \mathcal{D}_0(\eta,\eta'),
    \qquad
    \forall \eta,\eta'\in\Lambda_{\mathsf{DC}}.
    \label{eq:app-D0-lipschitz}
\end{align}
If the diameter of $\mathcal{D}_0$ is larger than one, define
\begin{align}
    \mathcal{D}(\eta,\eta')
    \triangleq
    \min\{
        \mathcal{D}_0(\eta,\eta'),
        1
    \}.
    \label{eq:app-normalized-metric}
\end{align}
Then $\mathcal{D}$ is again a metric. This is due the fact that for any $a,b\geq0$, we have
\begin{align}
    \min\{a+b,1\}
    \leq
    \min\{a,1\}
    +
    \min\{b,1\}.
    \label{eq:app-min-triangle}
\end{align}
Indeed, if $a<1$ and $b<1$, then the right-hand side is
$a+b$, and the inequality is immediate. If either $a\geq1$
or $b\geq1$, then the right-hand side is at least one, while
the left-hand side is at most one.
Therefore, by this new construction,
\begin{align}
    \operatorname{diam}
    \left(
        \Lambda_{\mathsf{DC}},
        \mathcal{D}
    \right)
    \leq
    1.
    \label{eq:app-normalized-diameter}
\end{align}
Moreover, the Lipschitz condition is preserved. If
$\mathcal{D}_0(\eta,\eta')<1$, then
$\mathcal{D}(\eta,\eta')=\mathcal{D}_0(\eta,\eta')$, so
\eqref{eq:app-D0-lipschitz} gives the result. If
$\mathcal{D}_0(\eta,\eta')\geq1$, then
$\mathcal{D}(\eta,\eta')=1$. Since the utility is normalized,
$0\leq\mathsf{U}(\eta)\leq1$, we have
\begin{align}
    |\mathsf{U}(\eta)-\mathsf{U}(\eta')|
    \leq
    1
    =
    \mathcal{D}(\eta,\eta').
    \label{eq:app-lipschitz-large-distance}
\end{align}
Therefore, replacing $\mathcal{D}_0$ by $\mathcal{D}$ does not change
the validity of the Lipschitz assumption.

\subsection{Why the Zooming Dimension Appears}
\label{subsec:app-why-zooming}

Recall that for each  $r\in(0,1]$, the zooming dimension examines
the gap annulus
\begin{align}
    \Lambda_r
    =
    \left\{
        \eta\in\Lambda_{\mathsf{DC}}
        :
        \frac{r}{2}
        <
        \Gamma(\eta)
        \leq
        r
    \right\}.
    \label{eq:app-Lambda-r-repeat}
\end{align}
If $\Lambda_r$ can be covered by at most $C_zr^{-d_z}$ sets of
diameter at most $r/8$, then $d_z$ measures the effective
dimension of the region where the regret gap is of order $r$.
Smaller $d_z$ means fewer distinguishable near-optimal thresholds
and a faster regret rate. The following examples compute $d_z$ explicitly.

\begin{example}
\label{ex:app-dz-upper-one}
Let
\begin{align}
    X=[a,b]\subset\mathbb{R},
    \qquad
    \mathcal{D}(x,y)
    =
    \min\{L|x-y|,1\}.
    \label{eq:app-ex-upper-metric}
\end{align}
For any utility $\mu:X\to[0,1]$ that is Lipschitz with
respect to $\mathcal{D}$, the zooming dimension is at most
one. Indeed, for any $r\in(0,1]$,
\begin{align}
    X_r
    \subseteq
    X.
    \label{eq:app-ex-upper-subset}
\end{align}
Partition $[a,b]$ into intervals of Euclidean length at most
$r/(8L)$. Each such interval has $\mathcal{D}$-diameter at most
$r/8$. The number of intervals is at most
\begin{align}
    \left\lceil
        \frac{8L(b-a)}{r}
    \right\rceil
    \leq
    \left(8L(b-a)+1\right)r^{-1}.
    \label{eq:app-ex-upper-cover}
\end{align}
Therefore, the covering condition holds with $d_z=1$ and
$C_z=8L(b-a)+1$. Hence the exact zooming dimension is always
bounded above by one under the scaled absolute-value metric.
\end{example}

\begin{example}
\label{ex:app-linear-gap}
Let
\begin{align}
    X=[0,1],
    \qquad
    \mathcal{D}(x,y)=|x-y|,
    \label{eq:app-ex-linear-metric}
\end{align}
and let
\begin{align}
    \mu(x)=1-x.
    \label{eq:app-ex-linear-utility}
\end{align}
Then $\mu^*=1$ and
\begin{align}
    \Delta(x)=\mu^*-\mu(x)=x.
    \label{eq:app-ex-linear-gap}
\end{align}
For $r\in(0,1]$,
\begin{align}
    X_r
    =
    \left\{
        x\in[0,1]:
        \frac{r}{2}<x\leq r
    \right\}
    =
    \left(\frac{r}{2},r\right].
    \label{eq:app-ex-linear-annulus}
\end{align}
This interval has length $r/2$. It can be covered by at most
four intervals of length $r/8$. Therefore, the covering number is
bounded by a constant independent of $r$, so
\begin{align}
    d_z=0.
    \label{eq:app-ex-linear-dz}
\end{align}
\end{example}

\begin{example}
\label{ex:app-polynomial-gap}
Let
\begin{align}
    X=[0,1],
    \qquad
    \mathcal{D}(x,y)=|x-y|,
    \label{eq:app-ex-poly-metric}
\end{align}
and let, for $p\geq1$,
\begin{align}
    \mu(x)
    =
    1-\frac{x^p}{p}.
    \label{eq:app-ex-poly-utility}
\end{align}
Then $\mu$ is $1$-Lipschitz on $[0,1]$ because
\begin{align}
    |\mu'(x)|=x^{p-1}\leq1.
    \label{eq:app-ex-poly-lipschitz}
\end{align}
The gap is
\begin{align}
    \Delta(x)
    =
    \frac{x^p}{p}.
    \label{eq:app-ex-poly-gap}
\end{align}
The annulus $X_r$ is
\begin{align}
    X_r
    =
    \left\{
        x\in[0,1]:
        \frac{r}{2}
        <
        \frac{x^p}{p}
        \leq
        r
    \right\}.
    \label{eq:app-ex-poly-annulus-start}
\end{align}
Equivalently,
\begin{align}
    X_r
    =
    \left(
        \left(\frac{pr}{2}\right)^{1/p},
        (pr)^{1/p}
    \right].
    \label{eq:app-ex-poly-annulus}
\end{align}
Its Euclidean length is
\begin{align}
    |X_r|
    =
    p^{1/p}
    \left(
        1-2^{-1/p}
    \right)
    r^{1/p}.
    \label{eq:app-ex-poly-length}
\end{align}
A set of $\mathcal{D}$-diameter at most $r/8$ has Euclidean
length at most $r/8$. Therefore, the covering number is on the
order of
\begin{align}
    \frac{r^{1/p}}{r}
    =
    r^{-\left(1-\frac{1}{p}\right)}.
    \label{eq:app-ex-poly-covering}
\end{align}
Thus the exact zooming dimension is
\begin{align}
    d_z
    =
    1-\frac{1}{p}.
    \label{eq:app-ex-poly-dz}
\end{align}

\end{example}

\begin{example}
\label{ex:app-flat-exp}
Let
\begin{align}
    X=[0,1],
    \qquad
    \mathcal{D}(x,y)=|x-y|,
    \label{eq:app-ex-exp-metric}
\end{align}
and define the gap by
\begin{align}
    \Delta(0)=0,
    \qquad
    \Delta(x)=e^{-1/x},
    \quad x>0.
    \label{eq:app-ex-exp-gap}
\end{align}
Let $\mu(x)=1-\Delta(x)$. For small $r$, the annulus is
\begin{align}
    X_r
    =
    \left\{
        x:
        \frac{r}{2}
        <
        e^{-1/x}
        \leq
        r
    \right\}.
    \label{eq:app-ex-exp-annulus-start}
\end{align}
Solving $e^{-1/x}=r$ gives
\begin{align}
    x=\frac{1}{\ln(1/r)}.
    \label{eq:app-ex-exp-endpoint-one}
\end{align}
Solving $e^{-1/x}=r/2$ gives
\begin{align}
    x=\frac{1}{\ln(2/r)}.
    \label{eq:app-ex-exp-endpoint-two}
\end{align}
Thus
\begin{align}
    X_r
    =
    \left(
        \frac{1}{\ln(2/r)},
        \frac{1}{\ln(1/r)}
    \right].
    \label{eq:app-ex-exp-annulus}
\end{align}
The length of this interval is
\begin{align}
    |X_r|
    =
    \frac{\ln 2}{
        \ln(1/r)\ln(2/r)
    }.
    \label{eq:app-ex-exp-length}
\end{align}
Covering it by intervals of length $r/8$ requires on the order of
\begin{align}
    \frac{1}{r\ln^2(1/r)}
    \label{eq:app-ex-exp-cover}
\end{align}
sets. This is at most a constant times $r^{-1}$, so $d_z=1$ is
sufficient. However, for any $d<1$,
\begin{align}
    \frac{
        1/(r\ln^2(1/r))
    }{
        r^{-d}
    }
    =
    \frac{
        r^{d-1}
    }{
        \ln^2(1/r)
    }
    \to
    \infty
    \quad
    \text{as }r\downarrow0.
    \label{eq:app-ex-exp-lower}
\end{align}
Therefore no $d<1$ is sufficient, and the exact zooming
dimension is
\begin{align}
    d_z=1.
    \label{eq:app-ex-exp-dz}
\end{align}
\end{example}

\begin{example}[Snowflake metric can give $d_z>1$]
\label{ex:app-snowflake}
Let
\begin{align}
    X=[0,1],
    \qquad
    \mathcal{D}_a(x,y)=|x-y|^a,
    \qquad
    0<a<1.
    \label{eq:app-ex-snowflake-metric}
\end{align}
This is a valid metric. Let
\begin{align}
    \mu(x)=1-x.
    \label{eq:app-ex-snowflake-utility}
\end{align}
Then
\begin{align}
    |\mu(x)-\mu(y)|
    =
    |x-y|
    \leq
    |x-y|^a
    =
    \mathcal{D}_a(x,y),
    \label{eq:app-ex-snowflake-lipschitz}
\end{align}
so $\mu$ is $1$-Lipschitz with respect to $\mathcal{D}_a$.
The gap is $\Delta(x)=x$, so
\begin{align}
    X_r
    =
    \left(\frac{r}{2},r\right].
    \label{eq:app-ex-snowflake-annulus}
\end{align}
This annulus has Euclidean length $r/2$. A set with
$\mathcal{D}_a$-diameter at most $r/8$ has Euclidean diameter at
most
\begin{align}
    \left(\frac{r}{8}\right)^{1/a}.
    \label{eq:app-ex-snowflake-euclidean-size}
\end{align}
Therefore, the number of sets required is on the order of
\begin{align}
    \frac{r}{r^{1/a}}
    =
    r^{-\left(\frac{1}{a}-1\right)}.
    \label{eq:app-ex-snowflake-cover}
\end{align}
Hence
\begin{align}
    d_z
    =
    \frac{1}{a}-1.
    \label{eq:app-ex-snowflake-dz}
\end{align}
For example, if $a=1/3$, then $d_z=2$. 
\end{example}

\end{document}